\documentclass[11pt, a4paper]{article}
\usepackage{}

\usepackage{lineno}
\pdfoutput=1
\usepackage{graphicx}
\usepackage{ae}
\usepackage{amsmath}
\usepackage{amssymb}
\usepackage{latexsym}
\usepackage{url}
\usepackage{epsfig}
\usepackage{cite}
\usepackage{mathrsfs}
\usepackage{amsfonts}
\usepackage{amsthm}

\input amssym.def
\input amssym.tex

\def\qed{\hfill$\Box$\vspace{12pt}}

\usepackage{color}

\long\def\delete#1{}

\newcommand{\be}{\begin{equation}}
\newcommand{\ee}{\end{equation}}
\newcommand{\ben}{\begin{equation*}}
\newcommand{\een}{\end{equation*}}
\newcommand{\bea}{\begin{eqnarray}}
\newcommand{\eea}{\end{eqnarray}}
\newcommand{\bean}{\begin{eqnarray*}}
\newcommand{\eean}{\end{eqnarray*}}

\def\diag{{\rm diag}}

\newtheorem{thm}{Theorem}[section]
\newtheorem{cor}[thm]{Corollary}

\newtheorem{lem}[thm]{Lemma}
\newtheorem{prop}[thm]{Proposition}

\numberwithin{equation}{section}

\theoremstyle{definition}
\newtheorem{exam}[thm]{Example}

\theoremstyle{remark}

\title{Some physical and chemical indices of
clique-inserted lattices}

\author{Zuhe Zhang
\\
{\small Department of Mathematics and Statistics}\\
{\small The University of Melbourne}\\
{\small Parkville, VIC 3010, Australia}\\
{\small zhang.zuhe@gmail.com}}

\date{}

\begin{document}

\openup 0.5\jot
\maketitle
%\linenumbers

\begin{abstract}
The operation of replacing every vertex of
an $r$-regular  lattice $H$ by
a complete
graph of order $r$ is called \emph{clique-inserting}, and
the resulting lattice is called the
\emph{clique-inserted lattice} of $H$.
For any given $r$-regular lattice, applying
this operation iteratively, an infinite family of
$r$-regular lattices is generated. Some
interesting lattices
including the 3-12-12 lattice can be constructed
this way. In this paper, we recall the relationship
between the spectra of an $r$-regular lattice and that of
its clique-inserted lattice, and investigate the \emph{graph energy} and \emph{resistance
distance} statistics.
As an application, the asymptotic energy per vertex
and
average
resistance distance of the 3-12-12 and 3-6-24
lattices are computed. We also
give formulae
expressing the numbers of spanning trees and
dimer coverings of the $k$-th
iterated clique-inserted lattices in terms of
that of the original one. Moreover,
we show that new families of expander graphs
can be constructed from the known ones
by clique-inserting.

\bigskip


\bigskip

\end{abstract}

\section{Introduction}
In the study of lattice statistical physics,
one family of $2$-dimensional lattices that have received a lot
of attention are constructed by replacing
each vertex of $r$-regular lattices
with a complete graph of order $r$ such that each of the $r$ new
vertices corresponds to one of the incident edges. (To avoid triviality, we assume $r\geq 3$
throughout the paper.)
Such lattices include the martini
\cite{kn:SCR,kn:WFY2,kn:TW}, the 3-12-12
\cite{kn:SW,kn:WFY,kn:TW}, the 3-6-24 \cite{kn:GSF} and
the modified bath room lattices \cite{kn:TW}.
Following \cite{kn:ZCC}, the operation of transforming
each vertex of an
 $r$-regular graph to an $r$-clique
 (complete graph of order $r$) is called
 \emph{clique-inserting}, and
 the graph obtained this way is called the
 clique-inserted graph of the original graph.
From a given
$r$-regular lattice $H$, the operation of clique-inserting can
 also be performed, and the resulting lattice $C(H)$
 is called the \emph{clique-inserted lattice} of
the original lattice.

 Throughout this paper, we always assume that
$G$ denotes
an undirected simple graph.
 Note that in the language of graph theory,
  the
clique-insertion operation on a graph $G$
 can be described as taking the line graph of
 the subdivision graph of $G$. For any given
 regular lattice $H$,
 by iterating
 this operation, a set of
  hierarchical regular lattices, namely,
  iterated clique-inserted lattices
  can be obtained. Denote $\{C^k(H)\}_{k \geq 0}$ the sequence of clique-inserted lattices with $C^{0}(H) \equiv H$ and $C^{k+1}(H)=C(C^k(H))$.
Start with the hexagonal lattice, the 3-12-12, 3-6-24 and
  3-6-12-48 lattices (refer to \cite{kn:GSF} for
  definitions of these lattices) can be generated
  by clique-inserting. In this case, the
  clique-insertion operation on a lattice is
equivalent to the fundamental ``Y-Delta" transformation (also known as the star-triangle transformation) on
 the subdivision graph of the original lattice.
 By this observation we obtained the relations between
 some physical and chemical indices of $r$-regular lattices
 and their $k$-th clique-inserted lattices.
 With such relations, we can compute some indices of some complex
 lattices easily based on
 the results of well-studied lattices such as the square and hexagonal
 lattices.

%It has been showed that, regardless of the structure of a $r$-regular graph $G$, the limits of the spectrum of $G$'s iterated clique-inserted graphs form the same fractal \cite{ZCC}.
%The asymptotic behavior of the number of spanning tree and Kirchhoff index of iterated clique-inserted graphs have also been studied in \cite{YYZ}.\\
%Throughout this paper, we always assume that
%$G$ denotes
%an undirected simple  graph.
%Recall that for an $r$-regular graph $G$,
 %the \emph{clique-inserted-graph} of $G$ is
 %defined as the graph obtained by
%replacing each vertex of $G$ with a complete
%graph of order $r$.

In this paper, we consider the lattices produced
 by the operation of
clique-inserting on regular lattices with free,
cylindrical and toroidal boundary conditions.
We will discuss the \emph{energy per vertex}, \emph{average
resistance} (the \emph{Kirchhoff index} over the number of pairs of vertices) and the entropy of spanning-tree and dimer models of such lattices. We will also use
the operation of
clique-inserting  to construct new families
of expander graphs from known ones.

The dimer model on regular lattices has attracted
the attention of many physicists as well as
mathematicians. For some classical works, we refer
to \cite{kn:TF,kn:FM,kn:WFY}.
Cayley\cite{kn:CA} and Kirchhoff\cite{kn:KG} presented
the problem of enumeration of spanning
trees of graphs, and further work in statistical
physics has appeared in both physics and
mathematics
literature. For a good survey, the reader is
referred to \cite{kn:SW}.
On the basis of electrical network theory, the study of \emph{resistance distance} was initiated by
Klein and Randi\'{c}\cite{kn:KR}, and the related index
named ``Kirchhoff index'' was well studied in \cite{kn:GM,kn:XG}.
In 1930s, H\"{u}ckel proposed a method for finding approximate solution of the Schr\"{o}dinger equation
of a class of organic molecules, the so-called conjugated hydrocarbons. In the framework of this modelization, the total $\pi$-electron energy, can be approximated by the sum of the absolute values of eigenvalues of the molecular graphs under certain chemical-based conditions.
Gutman abstracted a mathematical notion from this application-driven analysis on molecular graphs, therefore he defined \emph{graph energy}
as a graph invariant\cite{kn:GI1,kn:GI2}.
Since then, graph energy has been studied extensively by Chemists and Mathematicians. Yan and Zhang \cite{kn:YZ} proposed the
energy per vertex problem for lattice systems and
showed that the energy per vertex
of 2-dimensional lattices is independent from the boundary conditions, under various
choices. For a comprehensive survey of
results and common proof methods obtained on graph energy, see the monograph on graph energy \cite{kn:LSG} and references cited therein.
\emph{Expander graphs} were first defined by Bassalygo and
Pinsker in the early 70's.
These graphs are regular sparse graphs with strong
connectivity properties, measured by vertex, edge
or spectral
expansion as described in \cite{kn:HLW}.
For a graph, having such a property has significant implications in various
disciplines including complexity theory, computer networks, statistical
mechanics and so on.
%For example, expansion is closely related to the
%convergence rates of Markov Chains and so it plays
% a crucial
%role in the study of Monte-Carlo algorithms in
%statistical mechanics \cite{kn:HLW}.

The rest of the paper is organized as follows.
The expression of the energy
and Kirchhoff index of $k$-th iterated
clique-inserted lattices of regular lattices are
discussed in sections $2$ and $3$, respectively.
As an application, we compute the energy per vertex
and average Kirchhoff index of the 3-12-12 and
3-6-24 lattices.
In Section $4$, we show that, given $z_{H}$ as
the entropy of spanning trees of an $r$-regular lattice
$H$,
the entropy of spanning trees of $C^k(H)$ (the $k$-th iterated
clique-inserted graph of $H$) is given by
$r^{-k}(z_{H}+s_k(r)\ln{r(r+2)})$ where
$s_k(r)=(r/2-1)(r^{k}-1)/(r-1)$.
We will also show that when $H$ is cubic,
the free energy per dimer of $C^k(H)$
is $\frac{1}{3}\ln{2}$. In Section $5$, inspired by the Liu and Zhou's work \cite{kn:LZ},
we show that by applying clique-insertion operation iteratively on an expander family, new families of expander graphs can be obtained.
We propose clique-inserting as a modification to extend the size of computer networks, with their expansion properties being preserved to a certain degree.

\section{Asymptotic Energy}
Let $G=(V(G),E(G))$ be a graph with vertex set
$V(G)=\{v_1,v_2,\ldots,v_n\}$ and edge set $E(G)$.
The adjacency matrix of $G$, denoted by
$A(G)$, is the $n\times n$ symmetric matrix such that
 $a_{ij}=1$ if vertices $v_i$ and $v_j$
are adjacent and $0$ otherwise. Let $d_G(v_i)$ be the degree of vertex $v_i$ of $G$.
The \emph{Line graph} $L(G)$
of $G$, is the graph such that each vertex of $L(G)$
represents an edge of $G$ and two vertices of $L(G)$
are adjacent if and only if their corresponding edges
of $G$ share a common end vertex in $G$.
The \emph{subdivision graph} $S(G)$ of a graph $G$ is the
graph obtained by inserting a new vertex into every
 edge of $G$. It is easy to see that $C(G)=L(S(G))$.
The energy of a graph $G$ with $n$ vertices,
denoted by $\mathcal E(G)$, is defined by
$\mathcal E(G)=\sum_{i=1}^{n}|\lambda_i(G)|$,
where $\lambda_i(G)$'s are the eigenvalues of the
adjacency matrix of $G$. The asymptotic energy per
vertex of $G$ \cite{kn:YZ} is defined by
$\lim\limits_{|V(G)|\rightarrow\infty}\frac{\mathcal E(G)}{|V(G)|}$.

\begin{lem}
\cite{kn:YZ}
Suppose that $\{G_n\}$ is a sequence of finite simple
 graphs with
bounded average degree such that $\lim\limits_{n\rightarrow
\infty}|V(G_n)|=\infty$ and $\lim\limits_{n\rightarrow
\infty}\frac{\mathcal E(G_n)}{|V(G_n)|}=h\neq 0$. If $\{G_n'\}$ is a
sequence of spanning subgraphs of $\{G_n\}$ such that
$\lim\limits_{n\rightarrow \infty} \frac{|\{v\in V(G_n'):
d_{G_n'}(v)=d_{G_n}(v)\}|}{|V(G_n)|}=1$, then
$\lim\limits_{n\rightarrow \infty}\frac{\mathcal
E(G'_n)}{|V(G_n')|}=h$. That is, $G_n$ and $G_n'$ have the same
asymptotic energy.
\end{lem}
\begin{lem}
\cite{kn:ZCC}
Let $G$ be an $r$-regular graph with $n$ vertices
and $m$ edges. Suppose that the eigenvalues of
$G$ are $\lambda_{1}$ = $r \geq \lambda_{2} \geq ... \geq \lambda_{n}$. Then the eigenvalues of the clique-inserted graph $C(G)$
of $G$ are $\frac{r-2\pm \sqrt{r^2+4(\lambda_{i}+1)}}{2}$, $i=1,2,...,n$, besides $-2$ and $0$ each with
multiplicity $m-n$.
\end{lem}

From Lemma 2.2, we immediately obtain the following corollary.
\begin{cor}
Let $G$ be an $r$-regular $(r\geq 3)$ graph with $n$ vertices and eigenvalues $\lambda_{1} \geq \lambda_{2} \geq ... \geq \lambda_{n}$,
the energy of the clique-inserted graph of $G$ is
$$\mathcal E(C(G))=\sum_{i=1}^{n} \sqrt{r^2+4(\lambda_{i}+1)} +n.$$
\end{cor}
We will use this result to calculate the asymptotic
energy per vertex of the 3-12-12 lattice and
its clique-inserted lattice in the rest of this
section.

\subsection{3-12-12 lattice}

Our notation for the hexagonal lattices
follows \cite{kn:YZ}. The hexagonal lattices on a $n \times m$ torus,
denoted by $H^t(n,m)$, are illustrated in Figure 1, where $(a_1,a^*_1), (a_2,a^*_2), \ldots, (a_{m+1},a^*_{m+1})$, $(b_1,b^*_1), (b_2,b^*_2),\\ \ldots, (b_{n+1},b^*_{n+1})$ are edges in $H^t(n,m)$.
\begin{figure}[htbp]
  \centering
 \scalebox{0.3}{\includegraphics{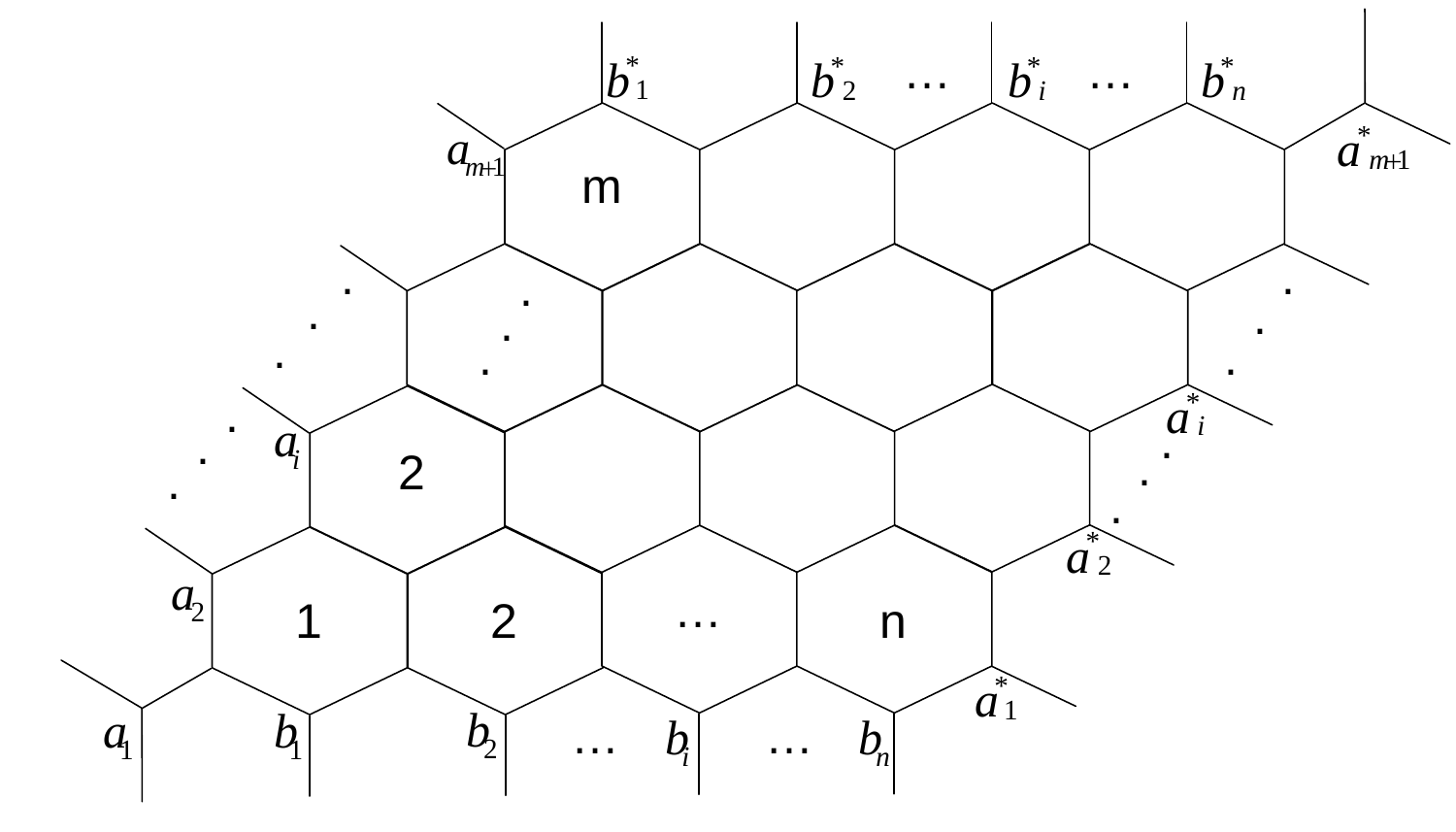}}
  \caption{$H^t(n,m)$ with toroidal boundary condition}
\end{figure}
%%and $(a_1,b_1), (a_2,b_2), \ldots, (a_{m+1},b_{m+1})$ are edges in
%%$H^c(n,m)$. If we delete edges $(a_1,b_1), (a_2,b_2), \ldots, (a_{m+1},b_{m+1})$ from $H^c(n,m)$, then the hexagonal lattice, denoted by $H^f(n,m)$, with free boundary condition is obtained (see Figure 3(c)).
%%In \cite, it is proved that the energy per vertex of $H^f(n,m)$, $H^c(n,m)$ and $H^t(n,m)$ are same.\\

By the definition of a clique-inserted lattice,
it is easy to see that each 3-12-12 lattice on the same geometry
is a clique-inserted-graph of $H^t(n,m)$, denoted
as $T^t(n,m)$ (see Figure 2(a)). Note that $(a_1,a^*_1), (a_2,a^*_2), \ldots, (a_{m+1},a^*_{m+1})$, $(b_1,b^*_1), (b_2,b^*_2), \ldots, (b_{n+1},b^*_{n+1})$ are edges in $T^t(m,n)$. If we delete edges $(b_1,b^*_1), (b_2,b^*_2),
\ldots, (b_{n+1},b^*_{n+1})$ from $T^t(n,m)$, then
the 3-12-12 lattice with cylindrical boundary
condition, denoted by $T^c(n,m)$ (see Figure 2(b))
 can be obtained. If we delete the edges
 $(a_1,a^*_1), (a_2,a^*_2),
\ldots, (a_{m+1},a^*_{m+1})$ from $T^c(m,n)$, then
the 3-12-12 lattice with free boundary condition,
denoted by $T^f(m,n)$ (see Figure 2(c)) can be
obtained.

\begin{figure}[htbp]
 \scalebox{0.6}{\includegraphics{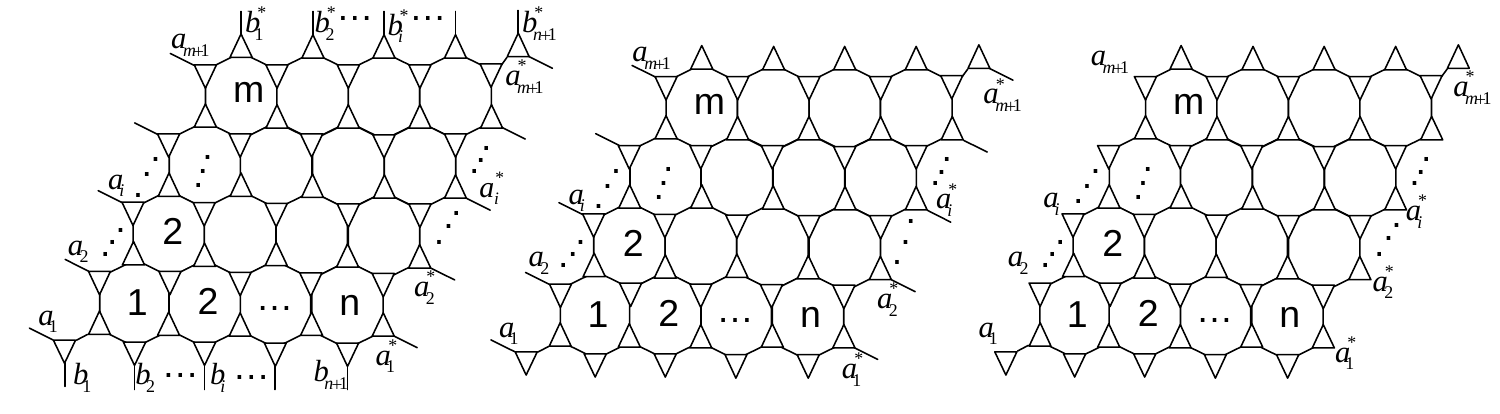}}
  \caption{\ The 3-12-12 lattice $T^t(n,m)$ (left), $T^c(n,m)$ (middle), and $T^f(n,m)$.}
\end{figure}

Note that almost all vertices of $T^c(m,n)$ and $T^f(m,n)$ are of degree $3$. Since $T^f(m,n)$ and $T^c(m,n)$
are spanning subgraphs of $T^t(m,n)$, by Lemma 2.1 we have
$$\lim_{n,m\rightarrow \infty}\frac{\mathcal E(T^t(n,m))}{6mn}
=\lim_{n,m\rightarrow \infty}\frac{\mathcal E(T^c(n,m))}{6mn}=\lim_{n,m\rightarrow \infty}\frac{\mathcal E(T^f(n,m))}{6mn}$$
It is shown in \cite{kn:YZ} that the eigenvalues of $H^t(n,m)$ are:\\ $$\pm\sqrt{3+2\cos\frac{2i\pi}{n+1}+2\cos\frac{2j\pi}{m+1}+2\cos\left(\frac{2i\pi}{n+1}+\frac{2j\pi}{m+1}\right)},
0\leq i\leq n, 0\leq j\leq m.$$
Since $T^t(n,m)$ is the clique-inserted graph of $H^t(n,m)$, we have
\begin{eqnarray*}
\mathcal E(T^t(n,m))&=&\sum\limits_{i=0}^{n}\sum\limits_{j=0}^{m}\sqrt{13+4\sqrt{3+2\cos\frac{2i\pi}{n+1}+
2\cos\frac{2j\pi}{m+1}+2\cos\left(\frac{2i\pi}{n+1}+\frac{2j\pi}{m+1}\right)}}\\
&&+\sum\limits_{i=0}^{n}\sum\limits_{j=0}^{m}\sqrt{13-4\sqrt{3+2\cos\frac{2i\pi}{n+1}+2\cos\frac{2j\pi}{m+1}+2\cos\left(\frac{2i\pi}{n+1}+\frac{2j\pi}{m+1}\right)}}+2mn\\
&=&\sum\limits_{i=0}^{n}\sum\limits_{j=0}^{m}\sqrt{26+2\sqrt{121-32\cos\frac{2i\pi}{n+1}-32\cos\frac{2j\pi}{m+1}-32\cos\left(\frac{2i\pi}{n+1}+\frac{2j\pi}{m+1}\right)}}+2mn.
\end{eqnarray*}
Thus, the average energy per vertex of 3-12-12 lattice can be expressed as
\begin{eqnarray*}
\lim_{n,m\rightarrow \infty}\frac{\mathcal E(T^t(n,m))}{6mn}&=&\frac{1}{3}+\frac{1}{24\pi^2}\int_0^{2\pi}\int_0^{2\pi}\sqrt{26+2\sqrt{121-32\cos x-
32\cos y-32\cos (x+y)}}{dx}{dy}\\
&=&1.4825....
\end{eqnarray*}
The last line follows by a numerical integration.
Therefore, the 3-12-12 lattices $T^t(n,m), T^c(n,m)$, and $T^f(n,m)$ with
toroidal, cylindrical, and free boundary conditions have the same
asymptotic energy ($\approx 8.895mn$).

\subsection{3-6-24 lattice}
\begin{figure}[htbp]
 \centering
\scalebox{0.5}{\includegraphics{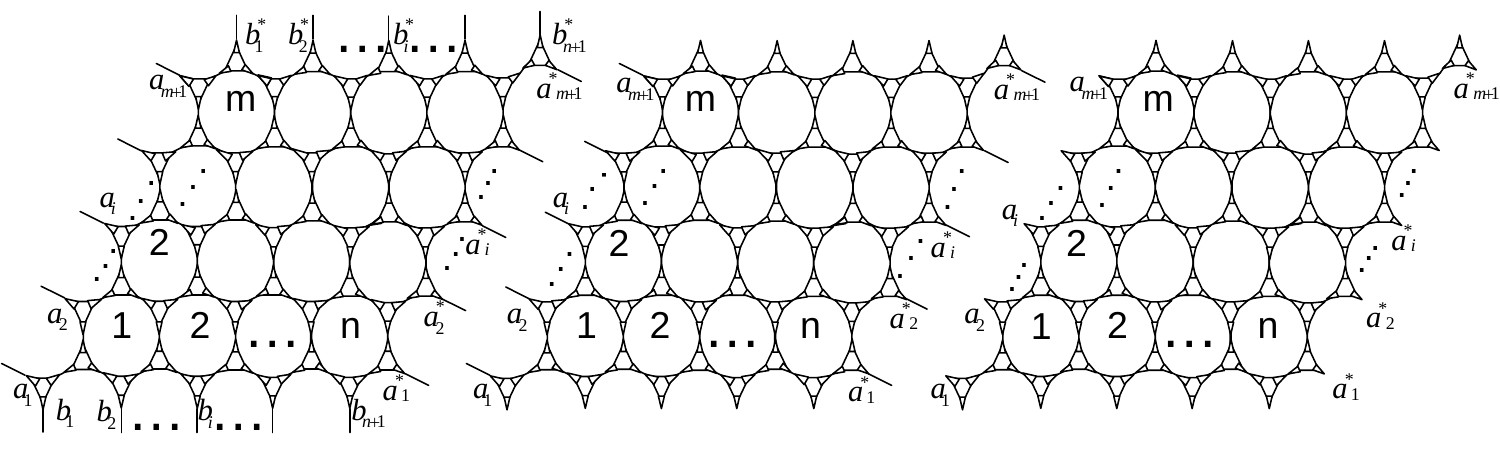}}
  \caption{\ The 3-6-24 lattice $S^t(n,m)$ (left), $S^c(n,m)$ (middle), and $S^f(n,m)$ .}
\end{figure}

The clique-inserted lattice of $T^t(m,n)$ is a lattice with toroidal boundary condition, denoted by $S^t(m,n)$, illustrated in Figure 3.
 Note that $(a_1,a^*_1), (a_2,a^*_2), \ldots, (a_{m+1},a^*_{m+1})$, $(b_1,b^*_1), (b_2,b^*_2),\\
 \ldots, (b_{n+1},b^*_{n+1})$ are edges in $S^t(m,n)$. If we delete edges $(b_1,b^*_1), (b_2,b^*_2),
\ldots, (b_{n+1},b^*_{n+1})$ from $S^t(n,m)$, then the 3-6-24 lattice with cylindrical boundary condition, denoted by $S^c(n,m)$ (see Figure 3(b)) can be obtained. If we delete edges $(a_1,a^*_1), (a_2,a^*_2),
\ldots, (a_{m+1},a^*_{m+1})$ from $S^c(m,n)$, then the 3-6-24 lattice with free boundary condition, denoted by $S^f(m,n)$ (see Figure 3(c)) can be obtained.

Note that $S^f(m,n)$ and $S^c(m,n)$
are spanning subgraphs of $S^t(m,n)$, by Lemma 2.1 we have
\begin{eqnarray*}
\lim_{n,m\rightarrow \infty}\frac{\mathcal E(S^t(n,m))}{18mn}
&=&\lim_{n,m\rightarrow \infty}\frac{\mathcal E(S^c(n,m))}{18mn}=\lim_{n,m\rightarrow \infty}\frac{\mathcal E(S^f(n,m))}{18mn}.
\end{eqnarray*}
The energy of the clique-inserted-graph of 3-12-12 lattice can be obtained by
\begin{eqnarray*}
%\mathcal E(S^t(n,m))&=&\sum\limits_{i=0}^{n}\sum\limits_{j=0}^{m}\sqrt{15+2\sqrt{13+4\sqrt{3+4\cos\frac{2i\pi}{n+1}+
%2\cos\frac{2j\pi}{m+1}+2\cos\left(\frac{2i\pi}{n+1}+\frac{2j\pi}{m+1}\right)}}}\\
%&&+\sum\limits_{i=0}^{n}\sum\limits_{j=0}^{m}\sqrt{15+2\sqrt{13-4\sqrt{3+4\cos\frac{2i\pi}{n+1}+
%2\cos\frac{2j\pi}{m+1}+2\cos\left(\frac{2i\pi}{n+1}+\frac{2j\pi}{m+1}\right)}}}\\
%&&+\sum\limits_{i=0}^{n}\sum\limits_{j=0}^{m}\sqrt{15-2\sqrt{13+4\sqrt{3+4\cos\frac{2i\pi}{n+1}+
%2\cos\frac{2j\pi}{m+1}+2\cos\left(\frac{2i\pi}{n+1}+\frac{2j\pi}{m+1}\right)}}}\\
%&&+\sum\limits_{i=0}^{n}\sum\limits_{j=0}^{m}\sqrt{15-2\sqrt{13-4\sqrt{3+4\cos\frac{2i\pi}{n+1}+
%2\cos\frac{2j\pi}{m+1}+2\cos\left(\frac{2i\pi}{n+1}+\frac{2j\pi}{m+1}\right)}}}\\
%&&+\sqrt{5}mn+\sqrt{13}mn+6mn.
\mathcal E(S^t(n,m))&=&\sum\limits_{i=0}^{n}\sum\limits_{j=0}^{m}\sqrt{30+2\sqrt{173-16\sqrt{3+2\cos\frac{2i\pi}{n+1}+
2\cos\frac{2j\pi}{m+1}+2\cos\left(\frac{2i\pi}{n+1}+\frac{2j\pi}{m+1}\right)}}}\\
&&+\sum\limits_{i=0}^{n}\sum\limits_{j=0}^{m}\sqrt{30+2\sqrt{173+16\sqrt{3+2\cos\frac{2i\pi}{n+1}+
2\cos\frac{2j\pi}{m+1}+2\cos\left(\frac{2i\pi}{n+1}+\frac{2j\pi}{m+1}\right)}}}\\
&&+\sqrt{5}mn+\sqrt{13}mn+6mn.
\end{eqnarray*}
Then the average energy per vertex of the
clique-inserted lattice of the 3-12-12 lattice is given by
\begin{eqnarray*}
\lim_{n,m\rightarrow \infty}\frac{\mathcal E(S^t(n,m))}{18mn}&=&\frac{1}{72\pi^2}\int_0^{2\pi}\int_0^{2\pi}\bigg( \sqrt{30+2\sqrt{173-16\sqrt{3+2\cos x+
2\cos y+2\cos (x+y)}}}\\
&&+\sqrt{30+2\sqrt{173+16\sqrt{3+2\cos x+2\cos y+2\cos (x+y)}}}\bigg) {dx}{dy}+\frac{\sqrt{5}+\sqrt{13}+6}{18}\\
&=& 1.4908....
\end{eqnarray*}
Thus, the lattices $S^t(n,m)$, $S^c(n,m)$, and $S^f(n,m)$ with
toroidal, cylindrical, and free boundary conditions have the same
asymptotic energy ($\approx 26.8344mn$).

\section{Average Resistance}
A graph can be viewed as an electrical network such that each edge of the graph is assumed to be a unit resistor. Then the resistance distance between vertices is defined as the effective resistance between them.
The Kirchhoff index $K(G)$ of a graph $G$ is defined
as the sum of the resistance distance between all pairs
of vertices of $G$. That is,
$$K(G)= \sum\limits_{\{u,v\}\subset{V(G)}}R(u,v)$$
where $R(u,v)$ denotes the resistance distance between vertices $u$ and $v$ of graph $G$.
Let $\overline{K}(G)=\frac{1}{{n\choose 2}}K(G)$ denote the average Kirchhoff index, that is, the average resistance distance between all pairs of vertices of $G$.
It has been showed in \cite{kn:LRH} that, in the large $n$ limit, the resistance distance between any two vertices $u$ and $v$ is dominated by the edges adjacent to $u$ and $v$ with contributions $1/d_{u}+1/d_{v}$. Therefore, the asymptotic average resistance of regular lattices are independent of the free, cylindrical and toroidal boundary conditions. Note that differently from the case of graph energy, deleting a cut-edge of a connected graph would change the resistance distance from finite to infinity.

\begin{lem}
\cite{kn:GLL} Let G be a connected $r$-regular graph with $n\geq2$ vertices. Then
$$K(L(G))=\frac{r}{2}K(G)+\frac{(r-2)n^2}{8}.$$
\end{lem}
\begin{lem}
\cite{kn:GLL}
Let $G$ be a connected $r$-regular graph with $n\geq2$ vertices. Then
$$K(S(G))=\frac{(r+2)^2}{2}K(G)+\frac{(r^2-4)n^2+4n}{8}.$$
\end{lem}
Combining the two lemmas above, the following result is straightforward.
\begin{prop}
Let $G$ be a connected $r$-regular graph with $n\geq2$ vertices. Then
$$K(C(G))=\frac{r(r+2)^2}{4}K(G)+\frac{(3r^3+2r^2-12r-8)n^2+8rn}{32}.$$
\end{prop}

\subsection{3-12-12 lattice}

It is shown in \cite{kn:Ye} that for the hexagonal
lattice $H^t(m,n)$ with
toroidal boundary condition,
\begin{eqnarray*}
K(H^t(m,n))\approx 10.9322(m+1)^2(n+1)^2.
\end{eqnarray*}
Therefore, by Proposition 3.3, we have
\begin{eqnarray*}
K(T^t(m,n))=204.9788(m+1)^2(n+1)^2+\frac{55(m+1)^2(n+1)^2+12(m+1)(n+1)}{8}.
\end{eqnarray*}
Thus, the asymptotic average resistance of the 3-12-12 lattice is given as follows:
\begin{eqnarray*}
\lim_{n,m\rightarrow \infty}\overline{K}(T^t(m,n))&=&\lim_{n,m\rightarrow \infty}\frac{K(T^t(m,n))}{{6(m+1)(n+1)\choose 2}}=11.7697....
\end{eqnarray*}

\subsection{3-6-24 lattice}

Based on the Kirchoff index of $T^t(m,n)$ and Proposition 3.3, we have
\begin{eqnarray*}
K(S^t(m,n))\approx 3843.3525(m+1)^2(n+1)^2+\frac{6105(m+1)^2(n+1)^2+1044(m+1)(n+1)}{32}.
\end{eqnarray*}
Thus, the asymptotic average resistance of the 3-6-24 lattice is given by,
\begin{eqnarray*}
\lim_{n,m\rightarrow \infty}\overline{K}(S^t(m,n))&=&\lim_{n,m\rightarrow \infty}\frac{K(S^t(m,n))}{{18(m+1)(n+1)\choose 2}}=24.9021....
\end{eqnarray*}

\section{Spanning Trees and Dimer Coverings}
\subsection{Spanning Trees}
Let $N_{ST}(G)$ denote the number of spanning trees of $G$. For $G$ which is a periodic lattice in finite dimension $D>1$, $N_{ST}(G)$ has asymptotic exponential growth.
Define the quantity $z_{G}$ by \\
$$z_{G}= \lim\limits_{n\rightarrow \infty}{\frac{1}{n}\ln{ N_{ST}(G)}}.$$
This quantity, corresponding to the free energy per site in the thermodynamic limit, is called \emph{bulk free energy}.
The following lemma indicates the relation between the number of spanning trees of a regular lattice and
of its $k$-th iterated clique-inserted lattice.
\begin{lem}
\cite{kn:YYZ}
Let $G$ be an $r$-regular graph with $n$ vertices. Then the number of spanning trees of the
iterated clique-inserted-graphs $C^k(G)$ of $G$ can be expressed by
$N_{ST}(C^k(G))=r^{ns-k}(r+2)^{ns+k}N_{ST}(G)$, where $s=s_k(r)=(r/2-1)(r^{k}-1)/(r-1)$.
\end{lem}
Therefore, we have the following proposition.
\begin{prop}
Let $H$ be an $r$-regular lattice. For $C^k(H) (k= 0, 1, 2, ..)$, the rate of growth of the number of spanning trees,
$z_{C^k(H)}$, is given by $r^{-k}(z_{H}+s\ln{r(r+2)})$,
where $s=(r/2-1)(r^{k}-1)/(r-1)$ and $z_{H}$ denotes the rate of growth of spanning trees of $H$.
\end{prop}
The next Theorem implies that the boundary condition does not affect the bulk limit of a lattice.
\begin{thm}
\cite{kn:Lyo05}
Let $\langle G_n \rangle$ be a tight sequence of finite connected graphs with bounded average degree
such that\\
$$\lim\limits_{n\rightarrow \infty}{\mid V(G)_{n}\mid^{-1}} \mid\{{x\in V(G'_{n}); deg_{G'_{n}}(x)=deg_{G_{n}}(x)}\}\mid=1,$$\\
then $\lim\limits_{n\rightarrow \infty}{\mid V(G)_{n}\mid^{-1}\log\ N_{ST}(G'_{n})=h}$.
\end{thm}
For the hexagonal lattice, $z_{hc}$ is 0.8076649... as shown in \cite{kn:SW}. Thus, by Proposition 4.2 and Theorem 4.3, we have that for the 3-12-12 and 3-6-24 lattices
with toroidal, cylindrical and free boundary condition,
$$z_{3-12-12}=0.7205633...$$
$$z_{3-6-24}=0.6915295....$$
\subsection{Dimer Coverings}
Let $M(G)$ denote the number of dimer coverings (perfect matchings) of $G$. The free energy per dimer of $G$, denoted by $Z_{G}$, is defined as
$Z_{G}=\lim\limits_{n\rightarrow \infty}{\frac{2}{n}\ln{M(G)}}.$
Given the number of vertices and edges of a connected graph, the number of dimer coverings of the graph and of its line graph have the following
relation.
\begin{lem}
\cite{kn:DYZ}
Let $G$ be a $2$-connected graph of order $n$ and size $m$, where $m$ is even and $\Delta(G)$ is the maximum degree of $G$. Then $M(L(G))\geq 2^{m-n+1}$, where the equality holds if and only if $\Delta(G)\leq 3$.
\end{lem}
With this general result, we can readily obtain the following.
\begin{prop}
Let $H$ be a cubic lattice with toroidal boundary conditions. The free energy per dimer of $C^k(H)$ ($k= 1, 2, 3, ..$) is equal to $\frac{1}{3}\ln{2}$.
\end{prop}
\begin{proof}
Assume that $H$ has $n$ vertices. Since $C^{k}(H)$ is the line graph of the subdivision of $C^{k-1}(H)$, by Lemma $4.4$ we have
$Z_{C^{k}}(G)=\lim\limits_{n\rightarrow \infty}\frac{2}{3^{k}n}\ln 2^{3^{k}n-\frac{5}{6}\cdot3^{k}n+1}= \frac{1}{3}\ln{2}$.\qed
\end{proof}

\begin{exam}
Let $R^{t}(m,n)$ be the $k$-th iterated clique-inserted lattice of the hexagonal
lattice $H^{t}(m,n)$ with toroidal boundary.
Note that the corresponding lattice $R^{c}(m,n)$ ($R^{f}(m,n)$) with cylindrical (free) boundary condition can be considered as the line graph of a graph which differs from $S(C^{k-1}(H^t{m,n}))$ by a small number (small in the sense that the number is $o(mn)$ as $m$,$n$ approach infinity) of edges. Therefore, by applying Lemma 4.4, we have
 $Z_{R^{t}(m,n)}=Z_{R^{c}(m,n)}=Z_{R^{f}(m,n)} =\frac{1}{3}\ln{2}$.\qed
\end{exam}
In general, when a cubic lattice is a line graph, the free energy per dimer of plane lattices are the same as that of the corresponding cylindrical and toroidal lattices. However, this may not be true when a cubic lattice is not a line graph. The hexagonal lattice is such a counterexample as shown in \cite{kn:yyz}.
\section{Expansion property}
Let $D(G)=\diag(d_G(v_1),d_G(v_2),\ldots,d_G(v_n))$ be
the diagonal matrix of vertex degree of $G$. The Laplacian matrix of
$G$ is $L(G)=D(G)-A(G)$. The eigenvalues of $L(G)$, denoted by $\mu_{1}\leq \mu_{2}\leq \cdots \leq \mu_{n}$ are called the Laplacian spectrum of $G$.
It is well known that $\mu_{2}$, called the \emph{algebraic connectivity} of $G$, is greater than $0$ if and only if $G$ is a connected graph. The \emph{spectral gap} of $G$ is defined as the difference of the largest and the second largest eigenvalues of $A(G)$. Note that for a regular graph, $\mu_{i}= r-\lambda_{i}$ for $i=1, 2, \ldots, n$, which implies that its spectral gap is equal to its algebraic
connectivity. Here we use spectral gap to quantify the expansion property, that is, a family of regular graphs is
an expander family if and only if there is a positive lower bound for their spectral gaps, and the larger the bound the better the expansion. This
characterization can be formulated to a formal definition as follows:\\
An infinite family of regular graphs, $G_{1}, G_{2}, G_{3},\ldots$, is called a family of $\varepsilon$-expander graphs \cite{kn:HLW},
where $\varepsilon > 0$ is a fixed constant, if (i) all these graphs are $k$-regular for a fixed integer $k \geq 3$;
(ii) $\mu_{2} \geq \varepsilon$ for $i = 1, 2, 3, \ldots$; and (iii) $n_{i} = |V (G_{i})|\rightarrow \infty$ as $i\rightarrow \infty$.
Note that Lemma 2.2 implies that\\
$$\mu_2(C(G))=\frac{r+2- \sqrt{(r+2)^2-4 \mu_{2}(G)}}{2}.$$ \\
Denote the function iteration of $f(x)=\frac{r+2- \sqrt{(r+2)^2-4x}}{2}$ by $f^1(x)=f(x)$ and $f^{k+1}(x)=f(f^{k}(x))$ for $k=1, 2, 3,\ldots$.

One primary application of expander graphs is in designing robust computer networks.
In the study of computer networks, it would be helpful to find simple and local graph operations to enlarge networks such that the new networks share similar topological properties with the old ones. For instance, Saad and Schultz studied the mapping which maps grid to hypercubes and found many topological properties are preserved under such an operation \cite{kn:SS}. In our case, applying clique-inserting on networks can be considered as replacing each workstation by a cluster (modeled by a complete graph) and rewiring them properly. By the following result, we will see that this provides a modest modification to enlarge the networks such that their expansion properties are maintained in some sense. 
\begin{prop}
Suppose $G_{1}, G_{2}, G_{3}, \ldots$, is a family of $r$-regular $\epsilon$-expander graphs. Then $C^k(G_{1}), C^k(G_{2}), C^k(G_{3}), \ldots,$
is a family of $r$-regular $f^k(\epsilon)$-expander graphs.
\end{prop}
\begin{raggedright}
Let $x=(\frac{2}{r+2})^2\epsilon$, then \end{raggedright}
\begin{eqnarray*}
 f(\epsilon)=\frac{(r+2)}{2}(1- \sqrt{1-x})
 &=&\frac{(r+2)}{2}\left(\frac{1}{2}x+\frac{1}{8}x^2+\frac{1}{16}x^3+\cdots\right)\approx \frac{\epsilon}{(r+2)}.
\end{eqnarray*}
This implies that the lower bound of the spectral gaps of the new expander family obtained by clique-inserting is a linear term of that of the original expander family. Note that it is simple and intuitive enough to perform realistic operations on networks according to clique-insertion. So even if the expansion properties of clique-inserted lattices are not exceptional, it is still meaningful to consider clique-insertion as an approach to extend computer networks, because in reality, the trade-off between performance and simplicity need to take into account.

%Lattices with good expansion property can be used to design the network with optimal synchronizability and fast random walk spreading .
Let us apply clique-inserting to the famous expander family $X^{p,q}$ of Lubostzky, Phillips and Sarnak \cite{kn:DSV}.
Recall that for a fixed real number $0<\gamma<1/6$ and sufficiently large $q$, the spectral gap of $X^{p,q}$ is bounded from below by $\varepsilon(r)=(p+1)-p^{\frac{5}{6}+\gamma}-p^{\frac{1}{6}-\gamma}$.
By Proposition $5.1$, for a fixed odd prime $p$, $C(X^{p,q})$ is a $(p+3- \sqrt{p^2+2p+4p^{\frac{5}{6}+\gamma}+4p^{\frac{1}{6}-\gamma}+5})/2$-expander family with degree $p+1$. More generally, $C^k(X^{p,q})$ is a $f^k((p+1)-p^{\frac{5}{6}+\gamma}-p^{\frac{1}{6}-\gamma})$-expander family.

\end{document}